\documentclass[preprint,12pt,nopreprintline]{elsarticle}

\usepackage[T1]{fontenc}
\usepackage[utf8]{inputenc}
\usepackage{lmodern}
\usepackage{amsmath,amssymb,amsfonts,amsthm}
\usepackage{braket}
\usepackage{graphicx}
\usepackage{booktabs}
\usepackage{array}
\usepackage{tabularx}
\usepackage{tikz}
\usepackage{microtype}
\usepackage{enumitem}
\usepackage{float}
\usetikzlibrary{arrows.meta,positioning,calc,fit}

\newcommand{\Ueff}{U_{\mathrm{eff}}}
\newcommand{\Reff}{R_{\mathrm{eff}}}
\newcommand{\pBSM}{p_{\mathrm{BSM}}}

\newcommand{\Tr}{\operatorname{Tr}}

\newtheorem{proposition}{Proposition}

\begin{document}

\begin{frontmatter}

\title{Rotational Bell-State Loop-Back Key Establishment and Agreement Certification for Passive Users}

\author[inst1]{Luis Adri\'an Lizama-P\'erez\corref{cor1}}
\ead{l.lizama@correo.ler.uam.mx}
\cortext[cor1]{Corresponding author}

\affiliation[inst1]{organization={Departamento de Sistemas de Informaci\'on y Comunicaciones, Divisi\'on de Ciencias B\'asicas e Ingenier\'ia, Universidad Aut\'onoma Metropolitana, Unidad Lerma},
addressline={Av. de las Garzas No. 10, Col. El Pante\'on},
city={Lerma},
postcode={52005},
state={Estado de M\'exico},
country={Mexico}}

\begin{abstract}
Server-aided Bell-state networks with sequential local encoding are an established approach to mediated key distribution. Against this background, we study a Loop-Back architecture in which two source-free, detector-free users transform the same traveling qubit. Alice privately samples $r\in\mathbb{Z}_2^2$, prepares $\ket{\beta_r}$ over the complete Bell basis, retains one qubit, and uses $r$ as the reference for the returned pair. A deterministic Pauli-composition mode, $\Ueff=U_2U_1$, serves as a reference to prior serial-unitary protocols. The main extension instead uses the two same-axis settings $R(\pm\alpha)$ and a reference-versus-complement Bell test. Opposite rotations cancel, whereas equal rotations add, so
\[
P(C\mid\mathrm{disagreement})=0,\qquad
P(C\mid\mathrm{agreement})=\sin^2(2\alpha).
\]
At $\alpha=\pi/8$, agreement is certified with probability $1/2$, and the overall conclusive-key probability is $1/4$ for unbiased choices. The same event supplies one shared raw-key bit. For the Pauli mode, a one-qubit Pauli twirl hides either user's factor from an endpoint mediator that may prepare an arbitrary qubit--ancilla state and choose the final measurement, provided the other Pauli is uniform and the traveler is inaccessible between users. The rotational guarantee is narrower: common-sign privacy holds only for the prescribed Bell source and coarse-grained instrument. A polarization implementation can normalize the private reference to the singlet and use a beam splitter with verified two-photon detection; loss produces a third, discarded outcome. The proposal is theoretical and does not claim composable security or experimental validation.
\end{abstract}

\begin{keyword}
mediated quantum key establishment \sep
Loop-Back quantum communication \sep
Bell states \sep
passive quantum users \sep
rotation encoding \sep
agreement certification \sep
Pauli twirling
\end{keyword}

\end{frontmatter}

\section{Introduction}

Quantum key distribution (QKD) has evolved from prepare-and-measure protocols to entanglement-based, measurement-device-independent, and networked architectures in which security, deployment cost, and trust are redistributed among different nodes \cite{bennett2014bb84,pirandola2020,lo2012}. A related line of work asks how much quantum functionality must remain at an end user. Semi-quantum protocols restrict one or more participants to a small set of operations, while mediated semi-quantum key distribution (M-SQKD) allows two restricted users to establish a key with the assistance of a quantum server \cite{boyer2007,krawec2015}. These models are particularly relevant when detectors, synchronized sources, memories, or stabilized interferometers are undesirable at the network edge.

Reducing user-side hardware can also remove mechanisms used to verify the server or detect active probing. We therefore study a specific trust--resource trade-off rather than claim device-independent or composable QKD. Here a \emph{passive user} neither prepares nor measures the protocol carrier and requires no quantum memory; the user only applies a calibrated transformation to the traveling subsystem. This endpoint model is stricter than the broader notion of a restricted or semi-quantum user.

A direct precedent is the server-aided Bell-state network of Li \emph{et al.} \cite{li2005}. Their server prepares and measures Bell states, while two users encode sequentially and recover one another's operations from the announced composition. That serial Pauli mechanism predates the present work. Li's users, however, perform single-particle control measurements, and Bob prepares replacement decoys to test the inter-user link and a potentially dishonest server. Removing those functions changes the security model; it is not a cost-free hardware reduction. We therefore retain the Pauli construction only as a reference mode and state the resulting trust and active-probing boundaries explicitly.

The broader mediated literature also contains stronger adversarial models. Krawec showed that two restricted users can establish a key through an untrusted quantum server with an unconditional-security proof \cite{krawec2015}. Circular M-SQKD subsequently placed restricted users on a sequential transport path and proved security against an untrusted third party \cite{ye2023circular,du2023circular}. More recent work has reduced measurement requirements and investigated lightweight mediated architectures \cite{zhou2024measurementfree}. These protocols use different user operations and verification structures from the present Loop-Back construction, but they set an important benchmark: a claim of secrecy against a malicious mediator requires an explicit mechanism or proof, not merely algebraic ambiguity.

Within the Loop-Back family, the architectural development followed a different route. Early Loop-Back formulations centralized preparation and detection while simplifying the remote optical node \cite{lizama2025symmetry,lizama2025entropy}. The subsequent passive-user construction uses single-qubit BB84 states and local polarization rotations $R(\pm\pi/8)$; the serial transformation $\Reff=R_2R_1$ either cancels or moves the state toward a conjugate basis, leading to an intrinsic conclusive probability $P_{\mathrm{conc}}=1/4$ \cite{lizama2026pauliarxiv}. A later single-user Bell-state Loop-Back construction introduced the retained-half geometry and, crucially, allowed the source to choose the initial Bell state privately over the complete Bell basis \cite{lizama2026bellarxiv}. The private label acts as a hidden reference: the station that prepared the pair can invert the observed Bell transition, whereas a party lacking the initial label has no Bell-frame reference with which to interpret that transition. The Bell formulation also changes the physical information carrier because the traveling subsystem is locally maximally mixed and the encoded information resides in the two-qubit correlations recovered at the active station.

The present work combines these two Loop-Back developments. It preserves the private complete-basis Bell reference and restores the $R(\pm\pi/8)$ mechanism of passive-user single-qubit Loop-Back. Successive Bell-state rotations are already used in QPC, including the $R_y(0)/R_y(\pi)$ construction of Hou and Wu \cite{hou2025ry}. The contribution claimed here is narrower: a private Bell reference, two passive endpoints, symmetric small-angle rotations on one traveler, and a binary test are combined so that one conclusive event both certifies agreement and supplies the users with a common raw-key bit. Li \emph{et al.} do not define this cancellation/addition map or its dual relation-and-key interpretation; their announced full Bell displacement instead supports deterministic reconstruction of the users' Pauli encodings \cite{li2005}.

This positioning is also consistent with current optical-network research. Recent \emph{Optical Switching and Networking} studies address resilience and entanglement management in quantum networks \cite{tunc2025,njoku2025}, while implementation-level work on ETSI QKD key-delivery interfaces emphasizes that usable key services depend on network and key-management overhead beyond physical-layer generation rates \cite{lauterbach2026}. The present contribution is situated one layer below those network-control problems: it studies an access primitive that moves quantum preparation and detection away from resource-constrained endpoints.

The two modes support different security statements. In the Pauli mode, a one-qubit Pauli twirl hides each factor from an endpoint mediator that controls the initial qubit--ancilla state and final measurement, provided the other Pauli is uniform and no intermediate access is possible. In the rotational mode, correctness and common-sign privacy require the specified Bell preparation and binary instrument. Active access to an individual modulator remains outside both results.

\section{Position with respect to prior work}
\label{sec:prior}

\subsection{Sequential Bell networks and mediated restricted-user protocols}

Li \emph{et al.} \cite{li2005} is the closest precedent to the deterministic Pauli mode. Their server prepares the Bell pair and performs the final Bell measurement, while the users apply four local encodings. Step S1 fixes the carrier to $\ket{\phi^+}_{AB}$, and the users actively verify the channel: both perform single-particle control measurements, and Bob may replace the traveler with a decoy before the Bob--Carol link. These operations test the server and channel and are part of the security mechanism.

Loop-Back instead centralizes preparation and detection and restricts each remote user to a transformation module \cite{lizama2025symmetry,lizama2025entropy,lizama2026pauliarxiv}. Extending this endpoint model to two Bell-state users removes Li's local verification functions. The result is a different allocation of quantum resources, with source-free, detector-free, memory-free users but a stricter trust boundary. Sections~\ref{sec:pauli} and~\ref{sec:security} state what remains protected and what does not.

The earlier single-user Bell-state Loop-Back protocol \cite{lizama2026bellarxiv} generalizes a different layer of the construction. Instead of fixing the initial Bell reference, the source samples $r\in\mathbb{Z}_2^2$ and prepares
\begin{equation}
\ket{\beta_r}=(I\otimes P_r)\ket{\Phi^+}.
\label{eq:prior-hidden-reference}
\end{equation}
For any Pauli displacement $P_u$,
\begin{equation}
(I\otimes P_u)\ket{\beta_r}\sim\ket{\beta_{r\oplus u}},
\label{eq:bell-covariance-prior}
\end{equation}
so the same transition algebra holds across all four Bell references. If $r$ is sampled uniformly and withheld from a party that sees only the final Bell label $m=r\oplus u$, then $m$ is uniform for every fixed $u$. This hidden-reference property is useful against observers that do not know $r$, but it must not be confused with secrecy against the source itself: Alice knows $r$ and can invert the displacement. In the present two-user protocol, private Bell-reference blindness and private factorization of $\Ueff$ are therefore treated as separate mechanisms.

M-SQKD provides a second comparison class. In the original mediated formulation, two users with restricted quantum capability establish a key through a server that may be adversarial \cite{krawec2015}. Circular transport has subsequently been analyzed with unconditional-security claims and explicit key-rate bounds \cite{ye2023circular,du2023circular}. Zhou \emph{et al.} further studied a measurement-free mediated scheme based on single-particle states \cite{zhou2024measurementfree}. These results show that ``lightweight users'' is not by itself a novelty claim; the relevant questions are which operations are removed, which verification mechanism replaces them, and what adversarial model remains covered.

A third comparison class is quantum private comparison (QPC). Zhou \emph{et al.} use $d$-level GHZ states for multi-party semi-quantum comparison, whereas Huang \emph{et al.} use high-dimensional Bell states to compare the size relation of two classical users' inputs \cite{zhou2025sqpc,huang2026sqpc}. These protocols establish that lightweight users and high-dimensional entanglement are already available for relation testing, but their goals, operations, and disclosed relations differ from the binary agreement event studied here. Tong \emph{et al.} provide a further $d$-dimensional Bell-state comparison design \cite{tong2026sqpc}.

Rotation encoding is also established in QPC. Hou, Sun, and Zhang combine Bell states with local rotations \cite{hou2024rotation}; Hou and Wu use a retained Bell half and successive $R_y(0)$ or $R_y(\pi)$ operations for deterministic equality testing under a semi-honest third party \cite{hou2025ry}. Their convention is $R_y(\vartheta)=\exp(-i\vartheta Y/2)$. The Loop-Back rotation used below is
\begin{equation}
R(\theta)=e^{-i\theta Y}
=
\begin{pmatrix}
\cos\theta&-\sin\theta\\
\sin\theta&\cos\theta
\end{pmatrix}
=R_y(2\theta),
\label{eq:rotation-convention}
\end{equation}
so the angles should not be compared without accounting for the factor of two. The present protocol uses the symmetric pair $R(\pm\alpha)$ and deliberately obtains a one-sided conclusive relation rather than deterministic equality.

Table~\ref{tab:prior} summarizes the prior architectures. It is intentionally conservative: sequential Bell transport, circular routing, rotations, and restricted users are not claimed as individually new.

\begin{table}[H]
\centering
\caption{Closely related prior architectures and their relation to the present construction. Detailed security models remain those of the original references.}
\label{tab:prior}
\scriptsize
\setlength{\tabcolsep}{4pt}
\renewcommand{\arraystretch}{1.08}
\begin{tabularx}{\textwidth}{>{\raggedright\arraybackslash}p{2.25cm}>{\raggedright\arraybackslash}p{3.55cm}>{\raggedright\arraybackslash}p{4.15cm}>{\raggedright\arraybackslash}X}
\toprule
Work & Path and user operations & Mediator/readout and security mechanism & Principal relation to this work \\
\midrule
Li \emph{et al.} \cite{li2005} & Fixed $\ket{\phi^+}$; four user unitaries; user control measurements and Bob-side replacement decoys & Full BSM; control and decoy rounds test the server and links & Closest serial-unitary precedent; users are active security participants \\
Prior Bell Loop-Back \cite{lizama2026bellarxiv} & One user; private $\ket{\beta_r}$ over the Bell basis; one Pauli & Full BSM relative to secret $r$ & Source-private complete-basis Bell reference inherited here \\
Krawec \cite{krawec2015} & Mediated SQKD with restricted users & Potentially adversarial server; unconditional-security analysis & Benchmark for malicious-mediator security \\
Ye \emph{et al.} \cite{ye2023circular} & Circular mediated path with restricted users & Protocol checks and unconditional-security analysis & Shows that circular restricted-user key establishment is established \\
Zhou \emph{et al.} \cite{zhou2024measurementfree} & Single-particle scheme; measurement-free users with reordering/delay functions & Third-party preparation and Bell measurement; attack analysis & Benchmark for reducing user measurement requirements \\
Hou--Wu \cite{hou2025ry} & Retained Bell half; sequential $R_y(0)$ or $R_y(\pi)$; decoy measurements & Full Bell comparison under a semi-honest TP & Bell-rotation precedent with deterministic equality \\
\bottomrule
\end{tabularx}
\end{table}

Against these precedents, the Pauli mode below is used as a deterministic baseline. The claimed extension is the rotational mode: it transfers the Loop-Back cancellation/addition rule to the retained-half Bell geometry and uses a binary event for the dual task of agreement certification and raw-key extraction. Table~\ref{tab:tradeoff} later compares only the resource and performance consequences of the Loop-Back modes, avoiding duplication of the prior-work comparison.

\subsection{Loop-Back lineage}

Figure~\ref{fig:lineage} separates the two ingredients that converge in the present protocol. The passive-user single-qubit branch contributes the serial small-angle algebra $R_2R_1=R(\theta_1+\theta_2)$. The single-user Bell branch contributes the retained subsystem, the locally maximally mixed traveler, and the private complete-basis Bell reference $r$. The present two-user construction combines these ingredients: the deterministic Pauli mode exposes the distributed displacement, whereas the rotational mode restores the original cancellation/addition rule and coarse-grains the Bell readout to a relation test.

\begin{figure}[H]
\centering
\begin{tikzpicture}[>=Latex,font=\small]
\tikzset{stage/.style={draw,rounded corners,align=center,text width=5.4cm,minimum height=2.35cm,inner sep=7pt}}
\node[stage] (rot0) at (-3.7,1.9) {\textbf{Passive-user single-qubit Loop-Back}\\BB84-like carrier; two sequential users\\$R(\pm\pi/8)$ and $R_2R_1=R(\theta_1+\theta_2)$\\nonorthogonal intermediates; $P_{\rm conc}=1/4$};
\node[stage] (bell1) at (3.7,1.9) {\textbf{Single-user Bell Loop-Back}\\private $r\in\mathbb{Z}_2^2$; $\ket{\beta_r}$ over the full Bell basis\\one retained half; local Pauli encoding\\$\rho_B=I/2$; full BSM relative to private $r$};
\node[draw,rounded corners,align=center,text width=12.3cm,minimum height=2.85cm,inner sep=8pt] (present) at (0,-2.0) {\textbf{Present two-user Bell Loop-Back framework}\\[2pt]
\emph{Inherited structure:} private complete-basis Bell reference $r$, retained subsystem $A$, and one traveler visiting $B_1$ then $B_2$.\\[3pt]
\textbf{Mode I -- Pauli reference:} $U_{\rm eff}=U_2U_1$ with full Bell-displacement readout.\qquad
\textbf{Mode II -- rotational:} $R(\pm\alpha)$ with binary private-reference/complement readout.};
\draw[->,thick] (rot0.south) -- ($(present.north)+(-3.0,0)$);
\draw[->,thick] (bell1.south) -- ($(present.north)+(3.0,0)$);
\end{tikzpicture}
\caption{Architectural lineage. The construction joins the small-angle rotation algebra of passive-user single-qubit Loop-Back with the retained-half, complete-basis private reference of the single-user Bell extension.}
\label{fig:lineage}
\end{figure}
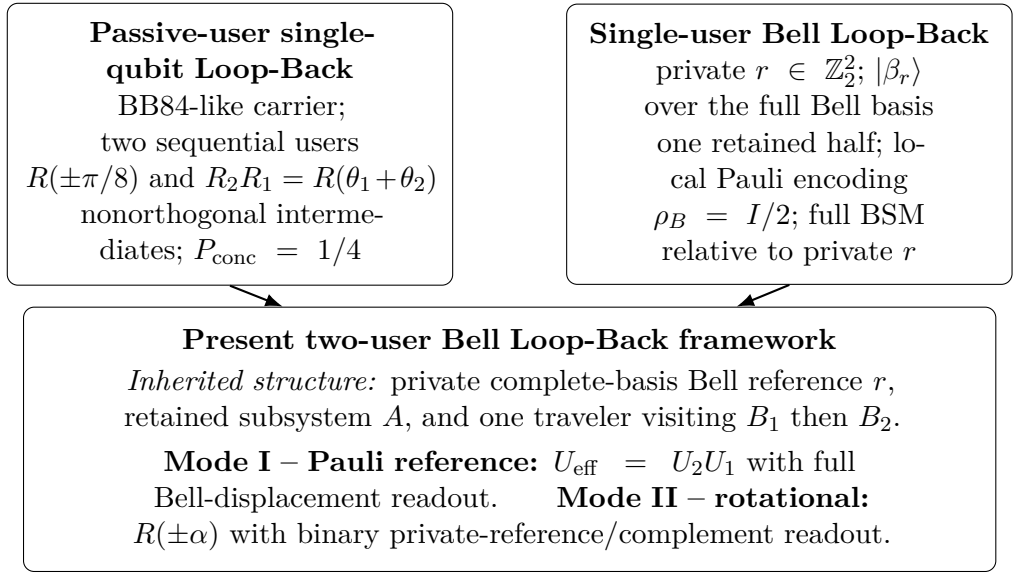

\section{Architecture and notation}
\label{sec:architecture}

Let
\begin{equation}
\ket{\Phi^+}=\frac{\ket{00}+\ket{11}}{\sqrt{2}},
\end{equation}
and label the Bell basis as
\begin{equation}
\ket{\beta_r}=(I\otimes P_r)\ket{\Phi^+},\qquad r\in\mathbb{Z}_2^2,
\label{eq:bell-label}
\end{equation}
where, modulo global phase,
\begin{equation}
P_{00}=I,\qquad P_{01}=Z,\qquad P_{10}=X,\qquad P_{11}=Y.
\label{eq:pauli-labels}
\end{equation}
Following the earlier single-user Bell-state Loop-Back construction \cite{lizama2026bellarxiv}, Alice samples $r$ uniformly from $\mathbb{Z}_2^2$, prepares $\ket{\beta_r}$, keeps $r$ private as a Bell reference, retains subsystem $A$, and sends subsystem $B$ through
\begin{equation}
A\longrightarrow B_1\longrightarrow B_2\longrightarrow A.
\label{eq:path}
\end{equation}
The same subsystem $B$ is therefore transformed twice before Alice measures the returned pair. For arbitrary user unitaries $U_1,U_2$,
\begin{equation}
\rho_{AB}\longmapsto
(I\otimes U_2U_1)\rho_{AB}(I\otimes U_1^\dagger U_2^\dagger).
\label{eq:serial}
\end{equation}
This serial action, rather than the particular choice of codebook, is the Loop-Back invariant used by both modes below.

For an honestly prepared Bell state, the traveler is locally maximally mixed:
\begin{equation}
\rho_B=\Tr_A\bigl(\ket{\beta_r}\!\bra{\beta_r}\bigr)=\frac{I}{2}.
\label{eq:maxmix}
\end{equation}
No local unitary applied by either user changes this marginal.

The private reference also gives a basis-independent covariance relation. If the total Pauli displacement has label $u$, then
\begin{equation}
(I\otimes P_u)\ket{\beta_r}\sim\ket{\beta_{r\oplus u}}.
\label{eq:bell-covariance}
\end{equation}
Writing the measured Bell label as $M=r\oplus u$, uniform secret $r$ gives
\begin{equation}
\Pr(M=m\mid u)=\frac14,\qquad I(U;M)=0,
\label{eq:hidden-reference-blindness}
\end{equation}
for a party that has access to the final Bell label but not to $r$. Alice is intentionally excluded from this blindness statement because she generated and retains $r$. Thus, Eq.~\eqref{eq:hidden-reference-blindness} formalizes hidden-reference blindness, whereas Section~\ref{sec:pauli} separately analyzes privacy of the individual user factors against Alice.

The covariance in Eq.~\eqref{eq:bell-covariance} can be seen directly in Table~\ref{tab:bell-transitions}. Global phases are omitted. Every Pauli operation permutes the entire Bell basis bijectively, so the inference rule used below is independent of which Bell state Alice selected as the private reference.

\begin{table}[H]
\centering
\caption{Action of a Pauli operation on the traveling qubit for all four possible Bell references. Global phases are omitted.}
\label{tab:bell-transitions}
\small
\setlength{\tabcolsep}{9pt}
\begin{tabular}{c c c c c}
\toprule
$U$ on traveler & $\ket{\Phi^+}$ & $\ket{\Phi^-}$ & $\ket{\Psi^+}$ & $\ket{\Psi^-}$ \\
\midrule
$I$ & $\ket{\Phi^+}$ & $\ket{\Phi^-}$ & $\ket{\Psi^+}$ & $\ket{\Psi^-}$ \\
$Z$ & $\ket{\Phi^-}$ & $\ket{\Phi^+}$ & $\ket{\Psi^-}$ & $\ket{\Psi^+}$ \\
$X$ & $\ket{\Psi^+}$ & $\ket{\Psi^-}$ & $\ket{\Phi^+}$ & $\ket{\Phi^-}$ \\
$Y$ & $\ket{\Psi^-}$ & $\ket{\Psi^+}$ & $\ket{\Phi^-}$ & $\ket{\Phi^+}$ \\
\bottomrule
\end{tabular}
\end{table}

Figure~\ref{fig:architecture} shows the common path, the private Bell reference, and the two readout choices.

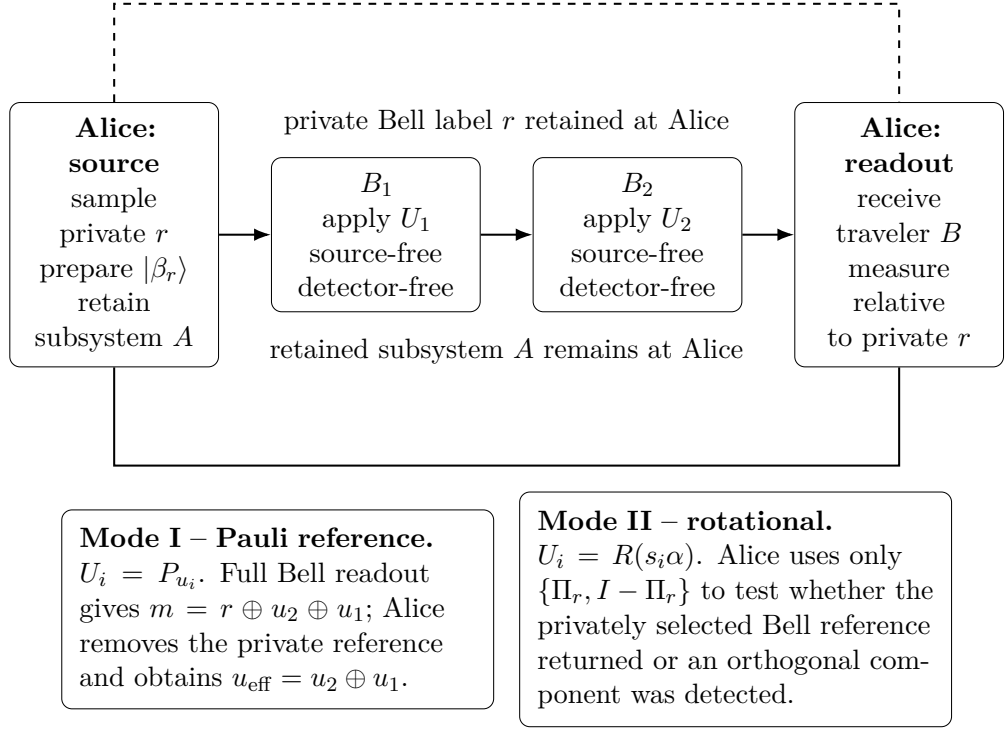
\begin{figure*}[t]
\centering
\resizebox{0.96\textwidth}{!}{%
\begin{tikzpicture}[>=Latex,font=\small]
\tikzset{
  op/.style={draw,rounded corners,align=center,text width=2.45cm,minimum height=1.45cm,inner sep=6pt},
  mode/.style={draw,rounded corners,align=left,text width=5.45cm,minimum height=2.15cm,inner sep=7pt}
}

\node[op] (a0) at (-5.4,2.0) {\textbf{Alice: source}\\sample private $r$\\prepare $\ket{\beta_r}$\\retain subsystem $A$};
\node[op] (b1) at (-1.8,2.0) {\textbf{$B_1$}\\apply $U_1$\\source-free\\detector-free};
\node[op] (b2) at (1.8,2.0) {\textbf{$B_2$}\\apply $U_2$\\source-free\\detector-free};
\node[op] (a1) at (5.4,2.0) {\textbf{Alice: readout}\\receive traveler $B$\\measure relative\\to private $r$};
\draw[->,thick] (a0.east) -- (b1.west);
\draw[->,thick] (b1.east) -- (b2.west);
\draw[->,thick] (b2.east) -- (a1.west);

\draw[dashed,thick] (a0.north) -- ++(0,1.35) -- ++(10.8,0) -- (a1.north);
\node[fill=white,inner sep=2pt] at (0,3.55) {private Bell label $r$ retained at Alice};

\draw[thick] (a0.south) -- ++(0,-1.35) -- ++(10.8,0) -- (a1.south);
\node[fill=white,inner sep=2pt] at (0,0.40) {retained subsystem $A$ remains at Alice};

\node[mode] (m1) at (-3.15,-3.15) {\textbf{Mode I -- Pauli reference.}\\
$U_i=P_{u_i}$. Full Bell readout gives $m=r\oplus u_2\oplus u_1$; Alice removes the private reference and obtains $u_{\rm eff}=u_2\oplus u_1$.};
\node[mode] (m2) at (3.15,-3.15) {\textbf{Mode II -- rotational.}\\
$U_i=R(s_i\alpha)$. Alice uses only $\{\Pi_r,I-\Pi_r\}$ to test whether the privately selected Bell reference returned or an orthogonal component was detected.};
\end{tikzpicture}%
}
\caption{Common retained-half path. Alice keeps subsystem $A$ and the private Bell label $r$ while the same traveler passes through the two transformation-only users. The Pauli mode resolves the displacement; the rotational mode uses a reference-versus-complement event.}
\label{fig:architecture}
\end{figure*}

\section{Mode I: deterministic Pauli composition}
\label{sec:pauli}

The Pauli mode supplies the deterministic baseline for the rotational construction. Let
\begin{equation}
U_i=P_{u_i},\qquad u_i\in\mathbb{Z}_2^2.
\end{equation}
Modulo global phase,
\begin{equation}
P_{u_2}P_{u_1}\sim P_{u_2\oplus u_1},
\end{equation}
so the returned Bell state is
\begin{equation}
(I\otimes P_{u_2}P_{u_1})\ket{\beta_r}
\sim \ket{\beta_{r\oplus u_1\oplus u_2}}.
\label{eq:pauli-bell}
\end{equation}
If the BSM returns Bell label $m$, Alice computes the displacement relative to her hidden reference,
\begin{equation}
u_{\rm eff}=m\oplus r=u_2\oplus u_1.
\label{eq:ueff}
\end{equation}
Thus the inference rule is valid for every initial Bell state; no Bell family is privileged. A party that knows $m$ but not $r$ cannot perform this inversion.
After Alice announces $u_{\rm eff}$, each user reconstructs the other's value:
\begin{equation}
u_2=u_{\rm eff}\oplus u_1,\qquad
u_1=u_{\rm eff}\oplus u_2.
\end{equation}
The ordered pair or either agreed component can then seed classical reconciliation and privacy amplification. This is a mediated key-establishment relation, not a claim that the complete protocol already satisfies a composable QKD definition.

\paragraph{Worked key-establishment round}
As a concrete round, let Alice privately choose $\ket{\Phi^-}$ as the Bell reference. Suppose $B_1$ applies $X$ and $B_2$ applies $Z$. Then $U_2U_1=ZX\sim Y$, and Table~\ref{tab:bell-transitions} gives
\begin{equation}
(I\otimes Y)\ket{\Phi^-}\sim\ket{\Psi^+}.
\end{equation}
Alice therefore observes $\ket{\Psi^+}$, removes her private reference, and announces only the effective displacement $Y$. From this announcement, $B_1$, who knows that it applied $X$, infers that $B_2$ applied $Z$; conversely, $B_2$ infers $X$. Both users thus reconstruct the same ordered pair $(X,Z)$, which can be treated as a shared raw symbol before the required classical post-processing. The example illustrates the reconstruction rule; the factor-privacy statement against Alice is the stronger ensemble result proved below.

\subsection{Factor privacy from Pauli randomization}

The four-fold degeneracy of Eq.~\eqref{eq:ueff} already hides each uniform Pauli factor from the public effective label. A stronger statement follows from Pauli randomization and directly addresses the role of the mediator.

\begin{proposition}[Pauli-twirling factor privacy]
\label{prop:twirl}
Let $R$ be an arbitrary ancilla controlled by Alice and let $B$ be a single-qubit traveler confined to the admitted optical mode. Alice may prepare any joint state $\rho_{RB}$ on this space, not necessarily a Bell state. Suppose $B_1$ applies a fixed Pauli $P_{u_1}$ and $B_2$ independently samples $u_2$ uniformly from $\mathbb{Z}_2^2$. If Alice has no access to $B$ between the two user operations, then the state returned to Alice, averaged over $u_2$, is
\begin{equation}
\rho^{(u_1)}_{RB,\rm out}
=\rho_R\otimes\frac{I_B}{2},
\label{eq:twirl-result}
\end{equation}
independent of $u_1$. The same statement holds with the users interchanged.
\end{proposition}

\begin{proof}
Conditioned on $u_1$ and averaged over $u_2$,
\begin{align}
\rho^{(u_1)}_{RB,\rm out}
&=\frac14\sum_{u_2}
(I_R\otimes P_{u_2}P_{u_1})\rho_{RB}
(I_R\otimes P_{u_1}^\dagger P_{u_2}^\dagger)\\
&=\frac14\sum_{v\in\mathbb{Z}_2^2}
(I_R\otimes P_v)\rho_{RB}(I_R\otimes P_v^\dagger)\\
&=\rho_R\otimes\frac{I_B}{2},
\end{align}
where the second line uses the fact that $u_2\mapsto v=u_2\oplus u_1$ is a permutation of the Pauli labels, and the last line is the one-qubit Pauli twirl.
\end{proof}

Consequently, no final measurement on the admitted qubit--ancilla system can reveal an individual factor when the other factor is an independent uniform Pauli. In information-theoretic notation,
\begin{equation}
I(U_1:RB_{\rm out})=I(U_2:RB_{\rm out})=0
\label{eq:factorMI}
\end{equation}
for the corresponding conditioned experiments. This is stronger than an honest-but-curious assumption on Alice's final measurement, but its scope is precise: it does not cover intermediate access, multiphoton or higher-dimensional probes, unmodeled optical degrees of freedom, or side channels that expose a modulator setting.

\section{Mode II: rotational Bell-state Loop-Back}
\label{sec:rotation}

The rotational mode restores the operation set used in the original passive-user single-qubit Loop-Back scheme. Each user selects
\begin{equation}
s_i\in\{-1,+1\},\qquad U_i=R(s_i\alpha),
\label{eq:rot-code}
\end{equation}
with $R(\theta)$ defined in Eq.~\eqref{eq:rotation-convention}. Since all rotations are about the same axis,
\begin{equation}
\Ueff=R(s_2\alpha)R(s_1\alpha)=R[(s_1+s_2)\alpha].
\label{eq:rot-compose}
\end{equation}
Thus
\begin{equation}
\Ueff=
\begin{cases}
I,&s_1\neq s_2,\\
R(+2\alpha),&s_1=s_2=+1,\\
R(-2\alpha),&s_1=s_2=-1.
\end{cases}
\label{eq:rot-cases}
\end{equation}
The decisive feature is therefore not a four-valued effective operation but a cancellation-versus-addition relation. Because the two settings differ only by the sign of a rotation about the same axis, the user-side alphabet also reduces from four Pauli settings in the reference mode to two calibrated rotation settings. This can simplify modulation control and calibration, although the physical component count remains implementation dependent.

For $\ket{\Phi^+}$,
\begin{equation}
(I\otimes R(\theta))\ket{\Phi^+}
=\cos\theta\ket{\Phi^+}+\sin\theta\ket{\Psi^-}.
\label{eq:phi-rotation}
\end{equation}
For a general Bell reference $\ket{\beta_r}$, conjugation by its Pauli label may change the sign of the orthogonal component but not the reference overlap. Hence
\begin{equation}
\left|\bra{\beta_r}(I\otimes R(\theta))\ket{\beta_r}\right|^2=\cos^2\theta.
\label{eq:overlap}
\end{equation}
Alice therefore need not resolve all four Bell labels. In the ideal model she applies the two-outcome projective instrument
\begin{equation}
\Pi_r=\ket{\beta_r}\!\bra{\beta_r},\qquad
\Pi_r^\perp=I_4-\Pi_r,
\label{eq:binary}
\end{equation}
and announces $N$ for the reference outcome or $C$ for the complement outcome. These probabilities assume a valid two-qubit measurement event. A lossy photonic receiver also produces a null outcome $\varnothing$ when two-photon occupancy cannot be verified. Such rounds are discarded before $N/C$ is interpreted, as detailed in Section~\ref{sec:implementation}.

\begin{proposition}[Conclusive agreement soundness]
\label{prop:soundness}
In the ideal rotational mode,
\begin{equation}
P(C\mid s_1\neq s_2)=0,
\qquad
P(C\mid s_1=s_2)=\sin^2(2\alpha).
\label{eq:soundness}
\end{equation}
Therefore a conclusive event implies agreement with zero false-positive probability.
\end{proposition}

\begin{proof}
For $s_1\neq s_2$, Eq.~\eqref{eq:rot-cases} gives $\Ueff=I$, so the returned state is exactly $\ket{\beta_r}$ and $\Pi_r^\perp$ never clicks. For $s_1=s_2$, the effective angle is $\pm2\alpha$ and Eq.~\eqref{eq:overlap} gives $P(N)=\cos^2(2\alpha)$, hence $P(C)=\sin^2(2\alpha)$.
\end{proof}

At the native Loop-Back setting $\alpha=\pi/8$,
\begin{equation}
P(C\mid\mathrm{agreement})=\frac12,
\qquad
P(C\mid\mathrm{disagreement})=0.
\label{eq:pi8}
\end{equation}
With independent unbiased user choices, $P(\mathrm{agreement})=1/2$, so
\begin{equation}
P(C)=\frac14.
\label{eq:keyprob}
\end{equation}
The $1/4$ rate is not presented as an efficiency gain over the original single-qubit rotational protocol; it is the same intrinsic conclusive probability, now realized in a Bell-correlation setting with a locally maximally mixed traveler.

\subsection{Dual use: certification and raw-key generation}

Assign the local bit mapping
\begin{equation}
+1\leftrightarrow 0,\qquad -1\leftrightarrow 1.
\end{equation}
When Alice announces $C$, Proposition~\ref{prop:soundness} guarantees $s_1=s_2$. Since each user already knows its own sign, both users hold the same bit
\begin{equation}
k=\frac{1-s_1}{2}=\frac{1-s_2}{2}.
\label{eq:keybit}
\end{equation}
Thus the same event has two consistent interpretations: Alice certifies the relation ``the user settings agree,'' while the users extract one common raw-key bit. No opening of the sign is required.

\paragraph{Worked certification round}
Consider a private consistency-check setting in which two users hold binary settings and require the mediator to certify only whether those settings coincide. Let Alice privately prepare $\ket{\Psi^-}$ with $\alpha=\pi/8$. For $s_1=s_2=+1$, the effective operation is $R(\pi/4)$, and
\begin{equation}
(I\otimes R(\pi/4))\ket{\Psi^-}
=\frac{1}{\sqrt{2}}\left(\ket{\Psi^-}-\ket{\Phi^+}\right).
\end{equation}
Alice tests only $\{\Pi_r,\Pi_r^\perp\}$. A conclusive outcome $C$ occurs with probability $1/2$; if it occurs, Alice announces $C$ and thereby certifies that the two private choices agree, without learning whether their common sign was $+1$ or $-1$ under the prescribed instrument. The users already know their own sign and, using Eq.~\eqref{eq:keybit}, both obtain the raw bit $k=0$ in this example. If their signs had been opposite, the rotations would cancel and $C$ could not occur; if an equal-sign round produces $N$, no certificate or raw-key bit is accepted from that round.

\begin{proposition}[Common-sign privacy under the prescribed instrument]
\label{prop:binaryprivacy}
Assume the specified Bell preparation and the binary instrument in Eq.~\eqref{eq:binary}. Conditioned on agreement, its classical outcome distribution is independent of whether the common sign is $+1$ or $-1$. Conditioned on either outcome, the normalized post-measurement state also carries no common-sign information. Hence the prescribed source-and-instrument model reveals the relation but not the common raw-key bit.
\end{proposition}

\begin{proof}
For an appropriately phased orthogonal Bell component $\ket{\gamma_r}$, the two agreement states can be written
\begin{equation}
\ket{\psi_\pm}=\cos(2\alpha)\ket{\beta_r}\pm\sin(2\alpha)\ket{\gamma_r}.
\end{equation}
Both signs give the same probabilities $\cos^2(2\alpha)$ and $\sin^2(2\alpha)$. After outcome $N$, both states project to $\ket{\beta_r}$; after outcome $C$, both project to $\ket{\gamma_r}$ up to a global sign.
\end{proof}

This proposition is deliberately instrument-specific. Before the binary measurement,
\begin{equation}
\left|\braket{\psi_+|\psi_-}\right|=|\cos(4\alpha)|.
\label{eq:agreement-overlap}
\end{equation}
At $\alpha=\pi/8$ the two agreement states are orthogonal. A mediator that replaces the prescribed binary instrument with a different measurement can therefore distinguish the two common signs perfectly in the ideal case. Rotational key privacy is consequently a \emph{measurement-constrained} property, not malicious-mediator security.

\subsection{Sensitivity to rotation error}

Let the intended rotations be $s_i\alpha$ but let the actual angles contain small errors $\delta_i$. For a disagreement pair, the ideal cancellation angle $0$ becomes
\begin{equation}
\Delta=\delta_1+\delta_2,
\end{equation}
so the false-certificate probability is
\begin{equation}
P_{\rm false}=\sin^2\Delta\simeq \Delta^2
\label{eq:false}
\end{equation}
for $|\Delta|\ll1$. For an agreement pair, the certificate probability becomes
\begin{equation}
P_C=\sin^2(\pm2\alpha+\Delta).
\label{eq:rot-error}
\end{equation}
These expressions give a direct calibration target and provide an analytic error model for the rotational mode.

\section{Security model and explicit boundaries}
\label{sec:security}

The purpose of this section is to state only the guarantees supported by the architecture. Table~\ref{tab:security} separates them from assumptions that remain external to the protocol.

\begin{table}[!t]
\centering
\caption{Security properties and boundaries of the two operating modes.}
\label{tab:security}
\small
\renewcommand{\arraystretch}{1.08}
\begin{tabularx}{\textwidth}{>{\raggedright\arraybackslash}p{3.05cm}>{\raggedright\arraybackslash}X>{\raggedright\arraybackslash}X}
\toprule
Threat or property & Pauli-composition mode & Rotational mode \\
\midrule
Observer of traveler only & For honest Bell preparation, sees $I/2$ & Same \\
Mediator chooses the qubit source and final measurement & Individual factor hidden by Proposition~\ref{prop:twirl}, within the admitted qubit/mode model, provided no intermediate access and the other Pauli is uniform & Not covered; correctness and common-sign secrecy require the specified Bell source and instrument \\
Mediator follows prescribed readout & Learns $u_{\rm eff}$ but not either uniform factor & Learns $C/N$ on valid events; $C$ certifies agreement without revealing the common sign \\
Access between $B_1$ and $B_2$ & Not covered by Pauli-twirl theorem & Not covered \\
Multiphoton, out-of-mode, or Trojan-horse probe & Not covered by Proposition~\ref{prop:twirl}; requires a protected local optical boundary & Same \\
Separable replacement of legitimate Bell traveler & Not covered by Proposition~\ref{prop:twirl} alone; dedicated Bell verification/test rounds are required & Same; binary data rounds alone are not a full intrusion test \\
Composable security against coherent/collective attacks & Not claimed & Not claimed \\
\bottomrule
\end{tabularx}
\end{table}

\subsection{Local blindness of an honest Bell resource}

Equation~\eqref{eq:maxmix} gives a simple architecture-level property. An external observer with access only to the legitimate traveler obtains no local information about the Bell label or any unitary applied before that observation, because
\begin{equation}
U\frac{I}{2}U^\dagger=\frac{I}{2}.
\end{equation}
This differs from the single-qubit Loop-Back setting, where intermediate nonorthogonality is the central physical limitation on discrimination. Bell encoding moves the relevant information into nonlocal correlations.

\subsection{Why the Pauli mode tolerates a stronger endpoint mediator model}

Proposition~\ref{prop:twirl} is the strongest privacy statement in this paper. Within the admitted single-qubit optical mode, Alice need not prepare the advertised Bell state or perform the advertised BSM. The other user's uniform Pauli acts as a quantum one-time pad on the returned qubit. The result still requires Alice to have no access between the two user operations and does not extend to additional photons, modes, or degrees of freedom. It therefore differs from the inter-user decoy of Li \emph{et al.} \cite{li2005}: it protects against arbitrary endpoint source/readout choices inside the qubit model, but not against internal or out-of-model probing.

\subsection{Why the rotational mode requires a stricter mediator assumption}

The rotational alphabet is only two-valued and does not form a one-qubit depolarizing twirl. Equation~\eqref{eq:agreement-overlap} therefore exposes an unavoidable boundary. Correctness and common-sign privacy are established for the specified Bell preparation and binary instrument; a mediator that changes the source or measurement is outside that result. In particular, at $\alpha=\pi/8$ the two equal-sign states are orthogonal before coarse-graining, so a deviating mediator can extract the common sign with an appropriate measurement. The rotational mode therefore requires a trusted, certified, or otherwise enforced source-and-measurement boundary.

This is the central security--complexity trade-off between the two modes. The Pauli mode provides stronger factor privacy at the cost of four local operations and complete Bell readout. The rotational mode provides a smaller user alphabet and binary readout, but it needs a stronger assumption on the mediator.

\subsection{Intermediate probing and the cost of detector-free users}

The absence of user-side detectors and state preparation is operationally attractive but removes active checking functions used in several earlier protocols. If Alice or Eve can inject a nonstandard probe into $B_1$, retrieve it before $B_2$, and perform a joint measurement with an ancilla, an individual unitary may be characterized as an oracle. Neither Proposition~\ref{prop:twirl} nor Proposition~\ref{prop:binaryprivacy} covers that situation.

The required implementation boundary is therefore explicit: the user module must accept only the legitimate optical mode and the network must prevent unauthorized extraction between the two operations. Practical QKD studies of Trojan-horse probing show why wavelength dependence, back-reflection, incoming optical power, and monitoring sensitivity must be treated as security parameters and motivate filtering, isolation, attenuation or input monitoring, and reflection control \cite{gisin2006trojan,jain2015risk}. The passive-user constraint does not require all such physical protections to reside inside $B_1$ or $B_2$. In an infrastructure-assisted or access-network realization, the private transformation module can be placed behind a short, characterized local optical boundary: passive spectral/temporal filtering and isolation may precede the modulator, while active power, wavelength, or timing monitoring can be implemented by the surrounding access infrastructure. The cryptographic user then remains source-free and detector-free even though the access boundary itself may contain active protection hardware.

This placement should be interpreted as architectural enforcement, not as a new protocol-level secrecy proof. It restricts direct adversarial probe access to the private transformation only under a trusted or certified access boundary; if that boundary is bypassed or controlled by a malicious infrastructure provider, neither Proposition~\ref{prop:twirl} nor Proposition~\ref{prop:binaryprivacy} supplies protection. Li-type user-side control measurements and decoy-state preparation provide a different route to active checking \cite{li2005}, but adding them at the users would abandon the source-free, detector-free passive-user constraint studied here. They are therefore treated as an alternative architectural point in the trust--resource space, not as countermeasures within the present passive protocol.

\section{Resource and performance trade-offs}
\label{sec:resources}

Four quantities must be kept distinct: intrinsic conclusive probability, analyzer success, valid two-photon detection, and end-to-end throughput.

In the single-qubit rotational Loop-Back mode, the ideal useful-round probability is
\begin{equation}
P_{\rm use}^{\rm SQ}=\frac14.
\end{equation}
In the deterministic Bell-Pauli mode, every successful Bell discrimination identifies $u_{\rm eff}$, so
\begin{equation}
P_{\rm use}^{\rm P}=\pBSM.
\end{equation}
With an ideal complete BSM, $\pBSM=1$; with passive linear optics and no additional nonvacuum resources, unambiguous discrimination of all four polarization Bell states is limited to $1/2$ \cite{calsamiglia2001}. Thus the improvement in intrinsic post-selection can be partially offset by the measurement apparatus.

For the rotational Bell mode,
\begin{equation}
P_{\rm use}^{\rm R}=P(C)=\frac12\sin^2(2\alpha)
\end{equation}
for unbiased user choices in the ideal instrument. At $\alpha=\pi/8$, $P_{\rm use}^{\rm R}=1/4$. If $p_{\rm valid}$ is the probability of verifying the returned two-photon event, the accepted conclusive probability is
\begin{equation}
P_{\rm acc}^{\rm R}=p_{\rm valid}\,\frac12\sin^2(2\alpha).
\end{equation}
All other detection patterns, including loss, are assigned $\varnothing$ and discarded. The Bell rotational mode therefore offers no intrinsic rate gain. Its potential benefit is a two-setting same-axis alphabet and a binary relation readout, which may reduce modulation/calibration and outcome-classification burden without guaranteeing fewer components or detectors.

Table~\ref{tab:tradeoff} makes the cost of each choice explicit.

\begin{table}[H]
\centering
\caption{Resource/performance trade-off within the Loop-Back family. Channel transmission and detector efficiency are omitted from the intrinsic probabilities.}
\label{tab:tradeoff}
\scriptsize
\setlength{\tabcolsep}{3.5pt}
\begin{tabularx}{\textwidth}{>{\raggedright\arraybackslash}p{2.7cm}>{\raggedright\arraybackslash}p{4.25cm}>{\raggedright\arraybackslash}p{2.65cm}X}
\toprule
Mode & Endpoint and Alice-side resources & Intrinsic useful event & Main cost or limitation \\
\midrule
Single-qubit rotational Loop-Back & Users: $R(\pm\pi/8)$ only. Alice: single-qubit preparation and measurement & $1/4$ & Non-orthogonal intermediate states; intrinsic post-selection \\
Bell-Pauli Loop-Back & Users: four Pauli settings, no source/detector. Alice: Bell source, retained half, full BSM & $\pBSM$ & Complete Bell discrimination; four-setting modulation; active-probing boundary \\
Rotational Bell Loop-Back & Users: $R(\pm\alpha)$ only. Alice: Bell source, retained half, binary Bell comparator & $\tfrac12\sin^2(2\alpha)$ intrinsically; multiplied by $p_{\rm valid}$ in practice & Trusted or enforced measurement; intrinsic post-selection and null detection events \\
\bottomrule
\end{tabularx}
\end{table}

An implementation-level rate should include channel and detector factors. For example, the Bell-Pauli raw information rate may be written schematically as
\begin{equation}
R_{\rm raw}^{\rm P}=2\,\pBSM\,p_{\rm ch}\,p_{\rm det},
\end{equation}
where the factor 2 is the two-bit Pauli label. This is a raw architectural throughput, not a secret-key-rate formula. A composable key rate requires an adversarial model, parameter estimation, error correction, and privacy amplification and is outside the claim made here.

\section{Photonic feasibility and network use}
\label{sec:implementation}

The proposal is theoretical; no experiment is reported. Its elementary components, however, correspond to standard photonic functions. Alice requires a Bell-pair source and a delay line or quantum memory long enough to retain one subsystem during the loop. The users require only a calibrated single-qubit polarization transformation. In the original Loop-Back convention, $R(\pm\pi/8)$ can be implemented by a polarization rotator or an electro-optic polarization modulator with two settings. Because both settings are rotations about the same axis and differ only in sign, this alphabet can reduce calibration and drive-control complexity relative to a four-setting Pauli implementation; the extent of any hardware reduction depends on the specific modulator.

The Pauli mode requires complete Bell-state information. Deterministic complete Bell analysis is not available with passive linear optics alone, although additional degrees of freedom or auxiliary resources can overcome the standard $50\%$ limit \cite{calsamiglia2001,schuck2006}. The rotational mode asks a strictly coarser measurement question. Alice needs only to decide whether the returned pair occupies her one-dimensional Bell reference or its orthogonal complement; she does not need to identify which Bell component in that complement was populated. Consequently, an optical realization can in principle be organized around fewer resolved outcome classes and simpler coincidence/event classification than a four-outcome Bell analyzer. This is a potential hardware advantage of both the rotational key-establishment and agreement-certification interpretations, since they use the same binary observable. We do not claim a universal reduction in detector or component count, which depends on the photonic encoding and analyzer design.

A concrete polarization-encoded realization follows directly from the private Bell reference. Under the labeling of Eq.~\eqref{eq:bell-label}, Alice can apply before detection the local Pauli normalization
\begin{equation}
Q_r=P_{r\oplus 11},
\label{eq:reference-normalization}
\end{equation}
so that
\begin{equation}
(I\otimes Q_r)\ket{\beta_r}\sim\ket{\Psi^-},\qquad
(I\otimes Q_r)\ket{\gamma_r}\sim\ket{\tau_r},
\label{eq:singlet-normalization}
\end{equation}
where $\ket{\gamma_r}$ is the orthogonal Bell component appearing in the rotational state of Section~\ref{sec:rotation} and $\ket{\tau_r}$ is one of the three symmetric triplet Bell states. The adaptive private-reference measurement is therefore unitarily reduced to the fixed singlet--triplet projective decomposition
\begin{equation}
\Pi_{\rm s}=\ket{\Psi^-}\!\bra{\Psi^-},\qquad
\Pi_{\rm t}=I_4-\Pi_{\rm s}.
\label{eq:singlet-triplet}
\end{equation}
Thus Alice's private label controls only a local premeasurement setting; the photodetection stage itself can remain fixed.

Two-photon interference at a balanced beam splitter supplies a natural implementation of this symmetry filter \cite{hong1987,mitchell2003}. For indistinguishable photons entering the two input ports, the antisymmetric singlet produces antibunching, whereas the symmetric triplet sector produces bunching. A verified separate-output coincidence realizes $N$, and a verified same-output two-photon event realizes $C$. Any record that does not verify two-photon occupancy---including one-photon loss or no click---is labeled $\varnothing$ and discarded. Photon-number-resolving detectors, or an equivalent spatial-splitting detector tree, are therefore useful. Temporal, spectral, polarization, and spatial mode matching remain necessary because imperfect indistinguishability reduces interference visibility and increases comparator error. Figure~\ref{fig:binary-photonic} summarizes the receiver.

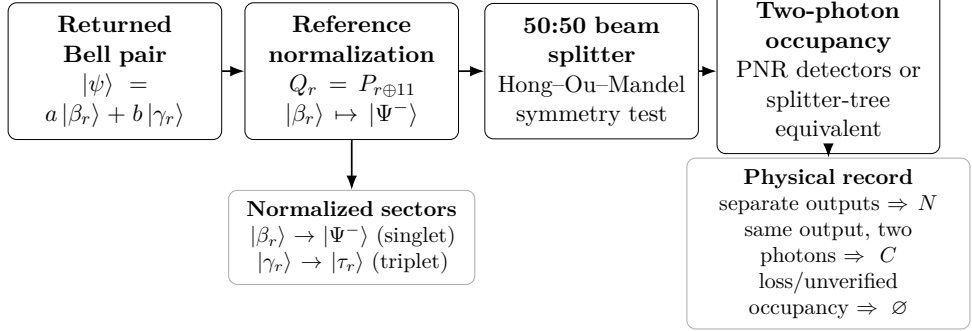
\begin{figure}[H]
\centering
\resizebox{0.98\textwidth}{!}{%
\begin{tikzpicture}[>=Latex,font=\small]
\tikzset{
  block/.style={draw,rounded corners,align=center,text width=3.05cm,minimum height=1.55cm,inner sep=6pt},
  smallblock/.style={draw,rounded corners,align=center,text width=3.20cm,minimum height=1.55cm,inner sep=6pt},
  flow/.style={->,thick},
  detail/.style={draw=gray!65,rounded corners,align=center,font=\footnotesize,text width=3.65cm,minimum height=1.52cm,inner sep=5pt},
  record/.style={draw=gray!65,rounded corners,align=center,font=\footnotesize,text width=4.25cm,minimum height=1.52cm,inner sep=5pt}
}
\node[block] (ret) at (-5.8,1.15) {\textbf{Returned Bell pair}\\$\ket{\psi}=a\ket{\beta_r}+b\ket{\gamma_r}$};
\node[block] (norm) at (-1.95,1.15) {\textbf{Reference normalization}\\$Q_r=P_{r\oplus 11}$\\$\ket{\beta_r}\mapsto\ket{\Psi^-}$};
\node[block] (bs) at (1.95,1.15) {\textbf{50:50 beam splitter}\\Hong--Ou--Mandel\\symmetry test};
\node[smallblock] (det) at (5.80,1.15) {\textbf{Two-photon occupancy}\\PNR detectors or\\splitter-tree equivalent};
\draw[flow] (ret) -- (norm);
\draw[flow] (norm) -- (bs);
\draw[flow] (bs) -- (det);

\node[detail] (classes) at (-1.95,-1.55) {\textbf{Normalized sectors}\\$\ket{\beta_r}\rightarrow\ket{\Psi^-}$ (singlet)\\$\ket{\gamma_r}\rightarrow\ket{\tau_r}$ (triplet)};
\node[record] (out) at (5.80,-1.65) {\textbf{Physical record}\\separate outputs $\Rightarrow N$\\same output, two photons $\Rightarrow C$\\loss/unverified\\occupancy $\Rightarrow\varnothing$};
\draw[flow] (norm.south) -- (classes.north);
\draw[flow] (det.south) -- (out.north);

\node[font=\footnotesize,align=center] at (0,2.65) {Alice's private label $r$ changes only the normalization; the symmetry test and detector stage remain fixed};
\end{tikzpicture}%
}
\caption{Conceptual photonic realization of the rotational Bell comparator. Alice uses the private Bell label $r$ to normalize the reference state to the antisymmetric singlet. A balanced beam splitter converts singlet--triplet symmetry into antibunched/bunched occupancy. The lower boxes separate the normalized sectors from the physical record: verified separate-output coincidences implement $N$, verified same-output two-photon events implement $C$, and loss or unverified occupancy gives the discarded outcome $\varnothing$. Photon-number resolution, or an equivalent spatial-splitting detector tree, prevents bunching from being confused with one-photon loss.}
\label{fig:binary-photonic}
\end{figure}

Conditioned on a valid two-photon event, the receiver implements the intended binary $1+3$ decomposition (singlet versus triplet). The physical record is nevertheless ternary, $\{N,C,\varnothing\}$, because loss and ambiguous occupancy cannot be assigned to either projector. This does not evade the standard $50\%$ limitation by implementing a better complete BSM; it avoids complete Bell labeling altogether. High indistinguishability, synchronization, verified two-photon detection, and calibrated efficiencies remain necessary. Setting-dependent or adversarially controlled loss is outside the present security claim.

At the network level, detector-free endpoints are most relevant when quantum resources are intentionally centralized, for example in passive access, edge, mobile, or route-configurable architectures. Current OSN research emphasizes that useful quantum networking also depends on resilience, entanglement-resource allocation, and key-management interfaces \cite{tunc2025,njoku2025,lauterbach2026}. The present primitive can be viewed as an access-layer candidate feeding such a key-management plane. Its output would still require authenticated classical signaling, key buffering, reconciliation where applicable, and application-facing key delivery.

The agreement-certification interpretation also permits a non-key use. A conclusive event certifies that two private binary settings coincide without requiring either user to reveal the setting itself. This may support private configuration-consistency checks, matched-policy gating, or distributed control rules whose acceptance condition is equality of two hidden binary settings when the mediator is trusted to enforce the binary instrument. Importantly, the event does not distinguish $(+,+)$ from $(-,-)$ and therefore does not by itself certify an affirmative value such as mutual consent or dual approval. These applications are secondary to the key-establishment analysis but clarify why the rotational mode exposes a relation rather than the full effective operation.

\section{Discussion}
\label{sec:discussion}

Comparison with Li \emph{et al.} isolates the incremental result \cite{li2005}. Their serial Bell--Pauli reconstruction is the baseline, whereas the rotational mode combines a private complete-basis reference, strict passive endpoints, the $R(\pm\alpha)$ alphabet, and a binary readout. Its conclusive event has a dual role absent from Li's reconstruction procedure: Alice certifies only that the settings agree, while the users obtain their common sign as a raw-key bit. This is a functional extension rather than a claim that removing user hardware alone creates a new cryptographic mechanism.

The extension also imposes a sharper trust--complexity trade-off. Pauli randomization hides either factor from a final endpoint measurement within the admitted qubit model, but not from intermediate or out-of-mode probes. The rotational mode reduces the local alphabet and readout classes, yet its common-sign privacy depends on the prescribed source and instrument. Bell encoding provides a locally maximally mixed traveler and a retained nonlocal reference, not a universal rate gain: at $\alpha=\pi/8$ the intrinsic conclusive probability remains $1/4$, and a physical receiver adds null events. Composable security, adversarial-loss analysis, and active-probe protection remain open.

\section{Conclusions}

This work separates a known deterministic Bell--Pauli baseline from a rotational Loop-Back extension. The latter places $R(\pm\pi/8)$ inside a retained-half, private-reference Bell geometry. Opposite rotations cancel and equal rotations reach the reference complement with probability $1/2$; for unbiased choices the resulting conclusive probability is $1/4$. The same conclusive event certifies agreement to Alice and gives the users one common raw-key bit without opening its sign under the prescribed instrument.

The guarantees are deliberately mode specific. A uniform Pauli hides the other user's factor from an endpoint mediator within the admitted qubit model when intermediate access is excluded. Rotational common-sign privacy instead requires the specified source and binary measurement. Its two-setting alphabet and singlet--triplet comparator may simplify control and readout, but entanglement generation, verified two-photon detection, post-selection, a protected optical boundary, and classical key post-processing remain necessary. No composable-security or experimental-validation claim is made.

\section*{Statements and Declarations}

\subsection*{Funding}
No external funding was received for this work.

\subsection*{Competing interests}
The author declares no competing interests.

\subsection*{Author contributions}
Luis Adri\'an Lizama-P\'erez: conceptualization, methodology, formal analysis, investigation, writing--original draft, writing--review and editing.

\subsection*{Data availability}
No new experimental or external datasets were generated or analyzed. All results reported here are analytical.

\renewcommand{\bibfont}{\small}
\setlength{\bibsep}{0.62em}

\end{document}